\documentclass[conference]{IEEEtran}
\IEEEoverridecommandlockouts
\usepackage{epsf}
\usepackage{graphicx}
\usepackage[outdir=./]{epstopdf}
\usepackage{amsmath,bm}
\usepackage{amssymb}
\usepackage{epsfig,latexsym,amsmath,epsf,amssymb,amsfonts}
\usepackage{verbatim}

\usepackage[english]{babel}
\usepackage[utf8]{inputenc}
\usepackage{algorithm}
\usepackage[noend]{algpseudocode}

\usepackage{placeins}

\usepackage{epstopdf}
\usepackage{amsthm}
\usepackage{url}
\usepackage{cite}
\usepackage{caption}
\usepackage{subcaption}
\newtheorem{theorem}{Theorem}

\usepackage{authblk}
\usepackage{amscd,mathrsfs}
\usepackage[numbers,sort]{natbib}
\usepackage[withbib,all]{authorindex}
\usepackage{esint}
\usepackage{paralist}
\usepackage{color}
\usepackage{bm}
\usepackage{amsthm}
\usepackage{mathtools}
\usepackage{mathptmx} 
\usepackage{times} 
\usepackage{amssymb}
\usepackage{float}
\usepackage{dsfont}
\newtheorem*{definition}{Definition}

\begin{document}

\pagestyle{empty}

\title{Wireless Energy and Information Transfer in Networks with Hybrid ARQ}
\author[1]{Mehdi Salehi Heydar Abad \thanks{This work was in part supported by EC H2020-MSCA-RISE-2015 programme under grant number 690893, and the Egyptian National Telecommunications Regulatory Authority.}}
\author[1]{Ozgur Ercetin}
\author[2,3]{Tamer Elbatt}
\author[2,4]{Mohammed Nafie}
\affil[1]{Faculty of Engineering
and Natural Sciences, Sabanci University}
\affil[2]{Electronics and Communications Engineering Dept., Cairo University, Giza 12613, Egypt}
\affil[3]{Computer Science and Engineering Dept., The American University in Cairo, AUC Avenue, New Cairo 11835, Egypt}
\affil[4]{Wireless Intelligent Networks Center, Nile University}
\affil[ ]{\textit{\{mehdis,oercetin\}@sabanciuniv.edu,\ telbatt@ieee.org,\ mnafie@ieee.org}}

\maketitle

\newtheorem{lemma}{Lemma}
\newtheorem{corollary}{Corollary}
\thispagestyle{empty}

\begin{abstract}
In this paper, we consider a class of wireless powered communication devices using hybrid automatic repeat request (HARQ) protocol to ensure reliable communications. In particular, we analyze the trade-off between accumulating mutual information and harvesting RF energy at the receiver of a point-to-point link over a time-varying independent and identically distributed (i.i.d.) channel. The transmitter is assumed to have a constant energy source while the receiver relies, solely, on the RF energy harvested from the received signal. At each time slot, the incoming RF signal is split between  information accumulation and energy accumulation with the objective of minimizing the expected number of re-transmissions.  A major finding of this work is that the optimal policy minimizing the expected number of re-transmissions utilizes the incoming RF signal to either exclusively harvest energy or to accumulate mutual information. This finding enables achieving an optimal solution in feasible time by converting a two dimensional uncountable state Markov decision process (MDP) with continuous action space into a countable state MDP with binary decision space.
\end{abstract}

\section{Introduction}
In simultaneous wireless information and power transfer (SWIPT), the incoming RF signal is used for both energy harvesting and decoding of information bits. The concept was first introduced by Varshney in \cite{varshney2008transporting},  characterizing the rates at which energy and reliable information can be transferred over a single point-to-point noisy link. It was later extended for frequency-selective channels with additive white Gaussian noise (AWGN) in \cite{grover2010shannon}. In \cite{zhang2013mimo}, the authors examined separated and co-located information and energy receivers in a multiple-input multiple-output (MIMO) wireless broadcast system. Specifically, for the co-located receiver case, two practical designs are investigated, namely time-switching (TS) and power splitting (PS). In TS policies, the incoming RF signal is either entirely utilized for energy or information purposes, whereas in PS policies the incoming signal is divided into two streams; one stream being utilized for energy and the other for information. 

In \cite{liangliu}, the optimal PS policy at the receiver was characterized to balance various trade-offs between the maximum ergodic capacity and the maximum average harvested energy in a single-input-single-output system. In addition, the optimal TS policy at the receiver is characterized for a point-to-point link over a narrow band flat-fading channel in \cite{liu2013wireless}.

In inherently error-prone wireless communications systems, re-transmissions, triggered by decoding errors, have a major impact on the energy consumption of wireless devices. Hybrid automatic repeat request (HARQ) schemes are frequently used to reduce the impact of re-transmissions by controlling them using various channel coding techniques. Nevertheless, this reduction comes at the expense of extra processing energy associated with the enhanced error-correction decoders. A receiver employing HARQ encounters two major energy consuming operations: (1) sampling or Analog-to-Digital Conversion (ADC), which includes all RF front-end processing, and (2) decoding. The energy consumption attributed to sampling, quantization and decoding plays a critical role in energy-constrained networks which makes their study a non-trivial problem. The work in \cite{de2014performance} investigated the performance of HARQ over an RF-energy harvesting point-to-point link, where the power transfer occurs over the downlink and the information transfer over the uplink. The authors studied the use of TS when two HARQ mechanisms are used for information transfer; Simple HARQ (SH) and HARQ with Chase Combining (CC) \cite{harq}. Recently, \cite{maha} studies the performance of HARQ in RF energy harvesting receivers. Particularly, the receiver employs a specific time switching policy to either harvest energy or accumulate mutual information in order to minimize the number of re-transmissions. However, it does not consider an accurate model for the energy consumption of the receiver components.


In this work, we consider a point-to-point link where a transmitter employs HARQ to deliver a message reliably to the receiver. The receiver has no energy source, and thus, it relies on harvesting RF energy from the same signal bearing information. The channel is time-varying where the amount of energy harvested and information collected varies depending on the quality of the channel.  The receiver aims to split the incoming RF signal between energy harvesting and information decoding so that the expected number of re-transmissions is minimized.  We develop a novel Markovian framework to prove that the optimal policy is a TS policy. As a consequence of this finding, we convert a two dimensional uncountable state Markov decision process (MDP) with continuous action space into a countable state MDP with binary decision space, and thus enabling us to use value iteration algorithm (VIA) to obtain the minimum expected number of re-transmissions in feasible time. Through numerical results, we show that the optimal policy is not unique and propose three heuristic policies to achieve the same performance as VIA. 


\section{System Model}\label{SystemModel}
\subsection{Channel Model and Receiver Architecture}
Consider a point-to-point time varying wireless link between a transmitter-receiver pair. 
The wireless channel is modeled according to an i.i.d. two-state block fading model where the states are GOOD and BAD. Note that the two-state channel process is an approximation of a more general multi-state time varying channel, where each state of the channel supports a maximum transmission rate. Here, we employ two-state channel process due to its analytical tractability. Let $G_t\in\left\{0,\ 1\right\}$ be the state of the channel at time slot $t$ where  a BAD state is denoted by $0$ and a GOOD state is denoted by $1$. We let the probability that the channel is in a GOOD state be $\lambda$, i.e., $\mathds{Pr}\left[G_t=1\right]=\lambda$. Let $g_i$ be the instantaneous complex channel gain corresponding to state $i=0,1$. We assume that the channel state information (CSI) is neither available at the transmitter nor the receiver due to the high computational and energy costs of transmitting and receiving a pilot signal necessary for measuring the CSI.

We consider a communication scheme where the transmitter is connected to a power source with an unlimited energy supply. The receiver is equipped with separate rectifier circuit for EH and a transceiver for information decoding (ID), both connected to the same antenna. We consider a co-located EH and ID architecture in which the EH and ID circuits share the same antenna to enable a compact structure. The incoming RF signal is fed to the EH and ID circuits according to time switching (TS) and power splitting (PS) architectures. 


Time is slotted and each slot has a length of $N$ channel uses. We assume that $N$ is sufficiently large so that we can apply information theoretic arguments. The instantaneous achievable rate of the receiver is the maximum achievable mutual information between the output symbols of the transmitter and input symbols at the receiver. Let us denote the achievable rate of the receiver by $R(t)$ at time $t$. As $N\to \infty$, $R(t)$ approaches the Shannon rate, and it can be computed as:
\begin{align}
R(t) = \log(1 + P |g(t)|^2),
\end{align}
where $g(t)\in \left\{g_0, g_1\right\}$ is the channel gain at time $t$ and $P$ is the noise-normalized transmit power of the transmitter. Let  $R_1$ and $R_2$ be the achievable rates corresponding to channel states GOOD and BAD, respectively. In particular
\begin{align}
R_1 = \log(1 + P |g_1|^2),\label{R1R2a}\\
R_2 = \log(1 + P |g_0|^2).\label{R1R2b}
\end{align}

Since the instantaneous channel states are not known prior to transmission, for reliability, we employ an HARQ scheme based on mutual information, namely HARQ with incremental redundancy (IR) \cite{wicker1995error}. Let us denote a message of the transmitter by $W\in \left\{1,2,\ldots,2^{NC}\right\}$, where $C$ denotes the rate of the information. Every incoming transport layer message into the transmitter is encoded by using a mother code of length $MN$ channel uses. The encoded message, $\mathbf{x}$, is divided into $M$ blocks, each of length $N$ channel uses, with a variable redundancy and it is represented by $\mathbf{x}=[x^1,\ldots,x^M]$. Let us assume that $x^1$ is transmitted at $t_1$. If $x^1$ is successfully decoded, then the receiver sends a 1-bit, error-free, zero-delay, Acknowledgement (ACK) message, otherwise, the transmitter times out after waiting a certain time period. In case of no ACK received, the transmitter transmits $x^2$ at time slot $t_2$ and the receiver combines the previous block $x^1$ with $x^2$. This procedure is repeated until the receiver accumulates $C$ bits of mutual information or maximum blocks of information, $M$, is sent. We assume that, $M$ is chosen sufficiently large so that the probability of decoding failure, due to exceeding the maximum number of re-transmissions, is approximately equal to zero. With HARQ-IR scheme, after $r$ re-transmissions, the amount of accumulated mutual information at the receiver is  $\sum^{r}_{k=1} R(t_k)$. The receiver, given that it has sufficient energy, can perform a successful decoding attempt after $r$ re-transmissions, if the amount of accumulated mutual information exceeds the information rate of the transmitted message, i.e., $\sum^{r}_{k=1} R(t_k)\geq C$. We assume that each message is encoded at rate $R_1$ i.e., $C=R_1$ so that a transmission  in a GOOD channel state carries  all the information needed for decoding

\subsection{Energy Harvesting and Consumption Model}
In the following, we assume that the receiver has a sufficiently large battery and memory, so that there is no energy or information overflow. The receiver utilizes a power splitting policy, where $\rho(t)\in [0,1]$ denotes the power splitting parameter at the beginning of time slot $t$. Note that $\rho(t)=0$ indicates that the received signal is used solely for mutual information accumulation, and $\rho(t)=1$ indicates that the received signal is used solely for harvesting energy. Any $0<\rho(t)<1$ refers to the case where the received signal is used for both harvesting energy and mutual information accumulation. Note that TS can be considered as a special case of PS with $\rho(t)\in \left\{0,1\right\}$.  

We incorporate a simplified EH model, which facilitates the formulation of a tractable optimization problem.  In this model, the receiver harvests a maximum of $e\geq 1$  energy units in the GOOD state and zero units during the BAD state\footnote{The maximum energy is harvested if the received signal is completely directed to the EH circuit, i.e., $\rho(t)=1$.}. The reason that no energy can be harvested during a BAD state is because in a typical EH device there are two stages: a rectifier stage that converts the incoming alternating current (AC) radio signals into direct current (DC); and a DC-DC converter that boosts the converted DC signal to a higher DC voltage value. The main limitation in an EH device is that every DC-DC converter has a minimum input voltage threshold below which it cannot operate. Hence, when the channel is in a BAD state, the input voltage is below the threshold of the DC-DC converter so no energy can be harvested. Albeit the receiver cannot harvest any RF energy in a BAD state, it can still accumulate mutual information since ID circuit operates at a lower power sensitivity,  e.g., $-10$ dBm for EH and $-60$ dBm for ID circuits \cite{EHsurvey}.

The energy consumption of HARQ was recently investigated in \cite{rosas2016optimizing}. The  energy is consumed at the start up of the receiver, during decoding, for operating passband receiver elements (low-noise amplifiers, mixers, filters, etc.), and for providing feedback to the transmitter.   In order to develop a tractable analytical model, we combine the individual costs of energy into two parameters only:  the receiver consumes $E_d\geq 1$ energy units for a decoding attempt and 1-energy unit for each mutual information accumulation event per time slot\footnote{One energy unit is normalized to the energy cost of operating the RF transceiver circuit during one time slot.}, i.e., operating the passband receiver elements.

\section{Expected Number of Re-transmissions}
\label{MCF}
The receiver requires at least $E_d$ units of energy and $R_1$ bits of information before it can successfully decode the transmitted packet. The objective is to optimally determine the power splitting ratio $\rho(t)$ between EH and ID so that the transmission is successfully decoded with minimum delay at the receiver. Note that $\rho(t)$ depends on the current battery level, $b$, and the amount of information accumulated, $m$.


\begin{definition}
A scheduling policy $\boldsymbol{\pi}= (\rho(1), \rho(2), \ldots,)$ is
a sequence of decision rules as such the $k$th element of $\boldsymbol{\pi}$
determines the power splitting ratio at $k$th time slot based on the observed system state $(b,\ m)$ at the beginning of this
time-slot for $t\in \{1,2,\ldots\}$. Similarly, a tail scheduling
policy $\boldsymbol{\pi}_t = (\rho(t), \rho(t+1), \ldots)$ is a sequence of decision rules
that determines the $\rho(t)$ for the time slots from $t$ to $\infty$.
\end{definition}

The problem can be mathematically modeled as a two-state Markov chain. Let the states of the Markov chain be  $(b,\ m)$, where $b$ is the total residual battery level and $m$ is the total accumulated mutual information normalized by $R_2$. For clarity of presentation, in the rest of the paper, we assume that $R_2=1$. 

\subsection{Dynamic Programming Formulation}

Let $f^{\boldsymbol{\pi}}(t)\in\{0,1\}$ be an indicator function taking
a value of $0$ if the message can be decoded at the end of slot
$t$ under policy $\boldsymbol{\pi}$, and a value of $1$ otherwise. Then,
the optimization problem we aim to solve is given as,
\begin{align}
\min_{\boldsymbol{\pi}} \sum_{t=0}^{\infty} f^{\boldsymbol{\pi}}(t).
\end{align}

Let $V^{\boldsymbol{\pi}}(b,m)$
be the expected discounted reward with initial state $S_0 =(0,0)$
under policy $\boldsymbol{\pi}$ with discount factor $\beta \in [0, 1)$. The expected discounted reward has the following expression
\begin{align}
V^{\pi}(b,\ m) = \mathds{E}^{\boldsymbol{\pi}}\left[\sum^{\infty}_{t=0}\beta^{t}U(S_{t},\rho(t))|S_{0}=(b,\ m)\right],\label{Vdef}
\end{align}
where $\mathds{E}^{\boldsymbol{\pi}}$ is the expectation with respect to the policy $\boldsymbol{\pi}$, $t$
is the time index, $\rho(t) \in [0,1]$ is the action chosen at time
$t$, and $U(S_{t},\rho(t))$ is the instantaneous reward acquired when
the current state is $S_t$.



In the rest of the paper, we use $\rho(t)$ and $\rho(b,m)$ interchangeably by assuming that at time slot $t$, the system is at state $(b,m)$. The battery is recharged with incoming RF signal depending on the value of the power split ratio $\rho(t)$.  Meanwhile, one unit of energy is consumed in order to accumulate non-zero bits of mutual information.  Hence, the evolution of the battery state is characterized as follows:
\begin{align}
B(t)=&\left\{
\begin{array}{ll}
B(t-1)+\rho(t) e-\mathds{1}_{\rho(t)\neq 1},& \text{if}\  G_t=1\\
B(t-1)-\mathds{1}_{\rho(t)\neq 1},& \text{if}\  G_t=0
\end{array}
\right.,
\end{align}
where $\mathds{1}_{\rho(t)\neq 1}=0$, if $\rho(t)=1$, and $\mathds{1}_{\rho(t)\neq 1}=1$, otherwise\footnote{When $\rho\neq1$, the receiver consumes 1 unit of energy to operate its transceiver.}.

According to (\ref{R1R2a}) and (\ref{R1R2b}), the transmit power is equal to $P = \frac{2^{R_1}-1}{|g_1|^2} = \frac{2^{R_2}-1}{|g_0|^2}$. At the power splitter, $1-\rho(t)$ portion of the received power is directed into the ID, so the maximum achievable mutual information accumulation is:
\begin{align}
R(t) = \log(1 + g(t)P(1-\rho(t)))\label{splittedR}
\end{align} 
Inserting the value of $P$ in (\ref{splittedR}) for GOOD and BAD channel states gives the mutual information accumulation in these states respectively for a given power splitting ratio $\rho$ as
\begin{align}
R^H(\rho)= \log(\rho+(1-\rho) 2^{R_1}),\\
R^L(\rho)= \log(\rho+(1-\rho) 2^{R_2}).
\end{align}

 Thus, the accumulated mutual information, $I(t)$, evolves as:
\begin{align}
I(t)=&\left\{
\begin{array}{ll}
\min(I(t-1)+R^H(\rho(t)),R_1),& \text{if}\  G_t=1\\
\min(I(t-1)+R^L(\rho(t)),R_1),& \text{if}\  G_t=0
\end{array}
\right..\label{I}
\end{align}

The instantaneous reward is zero if the message can be
correctly decoded, and it is minus one otherwise. Note that the
decoding operation is successful if and only if the accumulated
mutual information is above a certain threshold, and the
battery level is sufficient to decode the message. Hence, the
instantaneous reward is given as follows:
\begin{align}
U(S_{t},\rho(t))=&\left\{
\begin{array}{ll}
0,& \text{if}\  B_t\geq E_d,\ \text{and}\ I(t)\geq R_1,\\
-1,& \text{if}\  \text{otherwise}.
\end{array}
\right..
\end{align}
Define the value function $V(b,m)$ as
\begin{align}
V(b,\ m) &= \max_{\pi}V^{\pi}(b,\ m),\ \forall b\in[0,\infty),\ \forall m\in\left[0,\ R_1\right]\label{Vmax}.
\end{align}
The value function $V(b, m)$ satisfies the Bellman equation
\begin{align}
V(b, m) = \max_{0\leq \rho \leq 1}V_\rho(b,m)
\end{align}
where $V_\rho(b,m)$ is the cost incurred by taking action $\rho$ when
the state is $(b, m)$ and is given by
\begin{align}
V_{\rho}(b,m)=U((b,\ m),\rho)+\beta\mathds{E}\left[V(\acute{b},\ \acute{m})|S_{0}=(0,\ 0)\right], \label{actionvalue}
\end{align}
where $(\acute{b}, \acute{m})$ is the next visited state and the expectation is
over the distribution of the next state. The use of expected
discounted reward allows us to obtain a tractable solution, and
one can gain insights into the optimal policy when $\beta$ is close
to $1$. Then, one can apply VIA to obtain the optimal discounted reward. However, this problem suffers from the curse of dimensionality as it is a two dimensional uncountable state Markov decision process (MDP) with continuous actions at every state. Also, letting $\beta\rightarrow 1$, to approximate the average reward, extremely slows the algorithm to the point of infeasibility \cite{mehdigilbert}. Hence, in the following, we propose a novel approach to gain insights into the structure of the optimal policy.

\subsection{Absorbing Markov Chain Analysis}

The Markov chain describing the operation of our system is an {\em absorbing} Markov chain, where all states except those $(b, m)$ where $b\geq E_d$, and $m\geq R_1$ are transient states.  The absorbing states are those where the receiver has both sufficient energy and information accumulated to correctly decode. In an absorbing Markov chain, the expected number of steps taken before being absorbed in an absorbing state characterizes \emph{mean time to absorption}. Hence, mean time to absorption starting in a given transient state $(b,m)$ provides the number of re-transmissions until successful decoding starting from this state. It should be noted that the receiver is blind to the CSI before choosing the power splitting ratio. However, after it decides to sample the incoming RF signal for mutual information accumulation, based on the received power, the amount of the information in the sampled portion of the RF signal is revealed to the receiver.

In a finite absorbing chain, starting from a transient state, the chain makes a finite number of visits to some transient states before its eventual absorption into one of the absorbing states. Hence the mean time to absorption of the chain, starting from transient state $i$  initially, is the sum of the expected numbers of visits made to transient states. In the following, we perform first-step analysis, by conditioning on the first step
the chain makes after moving away from a given initial state to obtain the mean time to absorption. Let $k_{b,m}$ be the expected number of transitions needed to hit
an absorbing state when the Markov chain starts from state $(b,\ m)$. 

Let us first consider the trivial case when the battery has less than one unit of energy, i.e., $b<1$, in which case the receiver must harvest the incoming RF signal.  In this case, the mean time to absorption starting from an initial state $(b,m)$ is
\begin{equation} \label{kBless1}
\begin{split}
k_{b,m} &= 1+ \lambda k_{b+e,m}+(1-\lambda) k_{b,m} \\
 &=\frac{1}{\lambda} + k_{b+e,m},\hspace{1cm} \text{if}\ b<1.
\end{split}
\end{equation}

Note that in (\ref{kBless1}), one slot is needed to harvest energy, and depending on the channel state in that slot,  the battery state either transitions to $b+e$ or remains the same. Similarly, if the amount of accumulated mutual information is $R_1$, there is no point in further accumulating mutual information since the receiver has sufficient mutual information to decode the incoming packet. Hence, 
\begin{equation} 
\begin{split}
k_{b,m} &= 1+ \lambda k_{b+e,m} + (1-\lambda)k_{b,m} \\
 &=\frac{1}{\lambda} + k_{b+e,m},\hspace{1cm} \text{if}\ m= R_1.
\end{split}
\end{equation}

The following lemma plays an important role in establishing the structure of the optimal policy. 

\begin{lemma}
\label{lemma1}
For any $E_d-i\cdot e \leq b < E_d-(i-1)\cdot e$ such that $i=1,\ldots,E_d$, given that $m=R_1$, the mean time to absorption is given by, $k_{b,R_1} = \frac{i}{\lambda}$.
\end{lemma}
\begin{proof}
The proof is by induction. For the base case assume that the claim is true for $i=1$ such that $E_d-e\leq b < E_d$. Note that since $m=R_1$, the optimal decision is to use incoming RF signal only for harvesting energy, i.e., $\rho^*(b,\ R_1)=1$. Thus,
\begin{align}
k_{b,R_1} = 1 + \lambda k_{b+e,R_1}+ (1-\lambda) k_{b,R_1}.
\end{align}
For $E_d-e\leq b < E_d$, if the channel is GOOD then the Markov chain transitions into the absorbing state state $(b+e,\ R_1)$, so $k_{b+e,R_1}=0$. Hence, $k_{b,R_1}=\frac{1}{\lambda}$ and thus, the lemma holds for $i = 1$. In the induction step assume that the lemma is true for some $i=n$, i.e., $k_{b,R_1}=n/\lambda$ for $E_d-n\cdot e\leq b<E_d-(n-1)\cdot e$. The mean time to absorption for the case $n+1$ is:
\begin{align}
k_{b,R_1} =& 1 + \lambda k_{b+e,R_1}+ (1-\lambda)k_{b,R_1},\nonumber\\
&\hspace{0.5cm}\text{for}\hspace{0.2cm} E_d-(n+1)e\leq b<E_d-nl,
\end{align}
which reduces to $k_{b,R_1} = \frac{n+1}{\lambda}$ for $E_d-(n+1)\cdot e\leq b<E_d-n\cdot e$. Thus, the lemma holds by induction.
\end{proof}

We will use Lemma \ref{lemma1}  to show that the optimal policy minimizing the mean time to absorption \emph{does not}  split the incoming RF signal. In order to show this, let us define two tail policies $\boldsymbol{\pi}^i_t=(a_i,\boldsymbol{\pi}_{t+1})$, $i=S,D$ taking different actions, $a_i$ in the current slot, but following the same set of actions, $\boldsymbol{\pi}_{t+1}$ afterwards\footnote{Note that $(a_i,\boldsymbol{\pi}_{t+1})$ defines a tail policy obtained by concatenating action $a_i$ in the current slot with tail policy $\boldsymbol{\pi}_{t+1}$.}. Let policy $\boldsymbol{\pi}^S_t=(\rho, \boldsymbol{\pi}_{t+1})$ be a tail policy that always splits the incoming RF energy, i.e., $0<\rho<1$, except when $B(t)<1$ or $I(t)=R_1$, when it only harvests energy.
Assume that the state of the system is $(b,\ m)$ at time slot $t$. Then, the mean time to absorption for tail policy $\boldsymbol{\pi}^S_t$ is:
\begin{align}
k^{\boldsymbol{\pi}^S}_{b,m} = 1 + \lambda k_{b-1+\rho e,m+R^H(\rho)} + (1-\lambda) k_{b-1,m+R^L(\rho)},\label{k_pi}
\end{align}
where $k_{x,y}$ is the mean time to absorption of policy $\boldsymbol{\pi}_{t+1}$ beginning at state $(x,y)$. Note that with probability $\lambda$ the channel is in  GOOD state and hence $\rho\cdot e$ units of energy is harvested\footnote{We assume that the energy harvesting circuit is generating energy linearly proportional to the energy of the incoming RF signal.}. However, one unit of energy is spent by operating the transceiver to accumulate $R^H(\rho)$ bits of mutual information. Meanwhile, with probability $1-\lambda$ the channel is in  BAD state, and no energy is harvested, but the transceiver still consumes one unit of energy while accumulating $R^L(\rho)$ bits of mutual information.
Meanwhile, under tail policy $\boldsymbol{\pi}^D_t$ the RF signal is never split at time slot $t$, but rather, it is completely used for  mutual information accumulation except when $B(t)<1$ or $I(t) = R_1$ when it harvests energy only. One can calculate $k^{\boldsymbol{\pi}^D}_{b,m}$ as follows:
\begin{align}
k^{\boldsymbol{\pi}^D}_{b,m} = 1 + \lambda k_{b-1,R_1} + (1-\lambda)k_{b-1,m+1}.\label{kalpha}
\end{align}

\begin{theorem}
\label{TH_dis}
There exists an optimal time switching (TS) policy minimizing the number of re-transmissions until successful decoding which only harvests energy or accumulates information at an arbitrary time slot.
\end{theorem}
\begin{proof}
Assume that at time slot $t$ the system is at state $(b,\ m)$. Consider policy $\boldsymbol{\pi}^S$ which always chooses $0< \rho<1$. Hence, it follows that $R^H(\rho)< R_1$, $R^L(\rho)< 1$ and, from (\ref{I}), we have $I(t)\leq R_1$. Also, it is easy to verify that for any $b$, we have $k_{b,m_1}\leq k_{b,m_2}$ whenever $m_1\geq m_2$. Thus, a lower bound on $k^{\boldsymbol{\pi}^S}_{b,m}$ in (\ref{k_pi}) can be established as,
\begin{align}
k^{\boldsymbol{\pi}^S}_{b,m}\geq 1 + \lambda k_{b-1+\rho e,R_1} + (1-\lambda) k_{b-1,m+1}.\label{eq:bound}
\end{align}

Furthermore, since $b-1<b-1+\rho\cdot e<b-1+e$, from Lemma \ref{lemma1}, we know that $k_{b-1+\rho e,R_1}=k_{b-1,R_1}$. Hence, the lower bound in (\ref{eq:bound}) is exactly the same as $k^{\boldsymbol{\pi}^D}_{b,m}$ given in (\ref{kalpha}), i.e., $k^{\boldsymbol{\pi}^D}_{b,m}\leq k^{\boldsymbol{\pi}^S}_{b,m}$.
\end{proof}

Theorem \ref{TH_dis} proves that a time switching (TS) policy can achieve the minimum mean time to absorption. As a result, the state space of the discrete Markov chain associated with the optimal TS policy is  $b = 0,1,\ldots,\infty$\footnote{Note that in reality the capacity of the battery is limited to $B_{max}$, resulting in total $B_{max}\cdot R_1$ number of states. }, and $m = 0,1,\ldots,R_1$. Thus, we have converted an uncountable state MDP with continuous actions (i.e., $\rho(b,m)\in [0,1]$) into a countable state MDP with binary decisions (i.e., $\rho(b,m)\in \{0,1\}$). Hence, the curse of dimensionality is lifted from the problem and we can use VIA to obtain the optimal TS decisions at each state for the reduced problem. Also, since the number of states is reduced dramatically, we can choose $\beta\rightarrow 1$ to approximate the average reward instead of the discounted reward. Note that applying the VIA to the original problem is not possible in feasible time due to the extreme complexity of the problem originated from uncountable states.

The TS structure of the optimal policy also encourages us to propose simple heuristic policies which is suitable for EH devices lacking the necessary computation power. Hence, we propose three simple to implement heuristic policies utilizing the TS structure. These policies are as follows:

\begin{itemize}
\item Battery First (BF): the receiver harvests energy first until it acquires $m$ units of energy and then starts accumulating mutual information.
\item Information First (IF): the receiver always accumulates mutual information 
unless $B=0$ or $I=R_1$.
\item Coin Toss (CT): the receiver harvests energy when $B=0$ or $I=R_1$, while it accumulates mutual information when $B=\geq m+1$. Otherwise, it tosses a fair coin to choose between EH or ID.
\end{itemize}
In the following, we evaluate the performance of the optimal policy obtained by solving (\ref{actionvalue}) and compare the result to that of heuristic policies.

\section{Numerical Results}\label{Results}

In this section, we evaluate the minimum expected number of re-transmissions  by maximizing the value function defined in (\ref{actionvalue}) by the VIA and compare the values by those obtained by BF, IF and CT policies. To be able to approximate mean value by VIA, we choose $\beta = 1-10^{-17}$. We calculate the expected number of re-transmissions by Monte Carlo (MC) simulations. We run Monte Carlo simulations for $10^7$ iterations and evaluate the sample mean. 

Table \ref{son} summarizes the expected number of re-transmissions for $R_1 = 10$, $e=1$, $\lambda=0.5$ and $E_d=5$ with respect to $R_2$ associated with different policies. It can be seen from Table \ref{son} that all policies achieve with a very close approximation, the same expected number of re-transmissions for common system parameters. 

The effect of channel quality on the expected number of re-transmission for $R_2 = 5$, $R_1=10$, $E_d=5$ and $e=2$ with respect to $\lambda$ is summarized in Table \ref{son_vs_lamda}. As expected, it can be seen that the expected number of re-transmissions decreases as the channel quality improves. This is because as the channel quality improves, the probability of harvesting energy and accumulating $R_1$ bits of mutual information also increases. Again, it can be seen that all policies have the same performance independent of the value of $\lambda$.

The results presented in Table \ref{son} and \ref{son_vs_lamda} show that the optimal policy is not unique. To investigate this, we optimize $\rho$ values by VIA algorithm at each state $(b,m)$ for $R_1 = 5$, $R_2=2$, $e=2$, $\lambda=0.5$ and $E_d=5$ and represent the optimal $\rho$ values in Figure \ref{fig:BF} and \ref{fig:IF}. Note that Figure \ref{fig:BF} and \ref{fig:IF} that are obtained by VIA, happen to be exactly the same as the BF and IF policies, respectively, where \emph{black holes} represent absorbing states, \emph{blue squares} represent $\rho=1$, and \emph{red diamonds} represent $\rho=0$.  Optimality of both Figure \ref{fig:BF} and \ref{fig:IF} means that $V^{BF}(b,m)=V^{IF}(b,m)$ for every $b=0,1,\ldots$ and $m=0,\ldots,R_1$.
By comparing Figure \ref{fig:BF} and \ref{fig:IF} it can be seen that the optimal policy should harvest energy whenever $b=0$ or $m=R_1$, and it should accumulate mutual information whenever $b\geq E_d +1$. Also, choosing between $\rho=0$ or $\rho=1$ does not alter the minimum expected number of re-transmission, whenever $1\leq b\leq E_d$ and $0\leq m \leq R_1-1$, i.e., $V_{0}(b,m)=V_{1}(b,m)$ for $1\leq b\leq E_d$ and $0\leq m \leq R_1-1$. Consequently, BF, IF, CT, and optimal policies achieve the same minimum expected number of re-transmissions.




\begin{table}[t]
\centering
\caption{Expected number of re-transmissions for $R_1 = 10$, $e=1$ and $E_d=5$ vs. $R_2$}
\label{son}
\begin{tabular}{lllllllllll}
               & $R_2=1$  & $R_2=2$  & $R_2=3$  & $R_2=4$  & $R_2=5$     \\
VIA             & 15.9910 & 15.8103 & 15.6235 & 15.2490 & 14.4992  \\
BF             & 15.9938 & 15.8116 & 15.6259 & 15.2504 & 14.4999  \\
IF             & 15.9941 & 15.8143 & 15.6245 & 15.2508 & 14.4987  \\
CT             & 15.9966 & 15.8140 & 15.6266 & 15.2491 & 14.5020 
\end{tabular}
\end{table}

\begin{table}[t]
\centering
\caption{Expected number of re-transmissions for $R_1 = 10$, $R_2=5$, $e=2$ and $m=5$ vs. $\lambda$}
\label{son_vs_lamda}
\begin{tabular}{lllllllllll}
               & $\lambda=0.1$  & $\lambda=0.2$  & $\lambda=0.3$  & $\lambda=0.4$  & $\lambda=0.5$   \\
VIA             & 40.8904 & 20.7979 & 14.0320 & 10.5985 & 8.4989  \\
BF             & 40.8920 & 20.7962 & 14.0337 & 10.6002 & 8.4995  \\
IF             & 40.8978 & 20.7960 & 14.0331 & 10.5991 & 8.5002  \\
CT             & 40.8961 & 20.8006 & 14.0333 & 10.5973 & 8.4986 
\end{tabular}
\end{table}

\begin{figure}[ht]
  \centering
    \includegraphics[scale=.5]{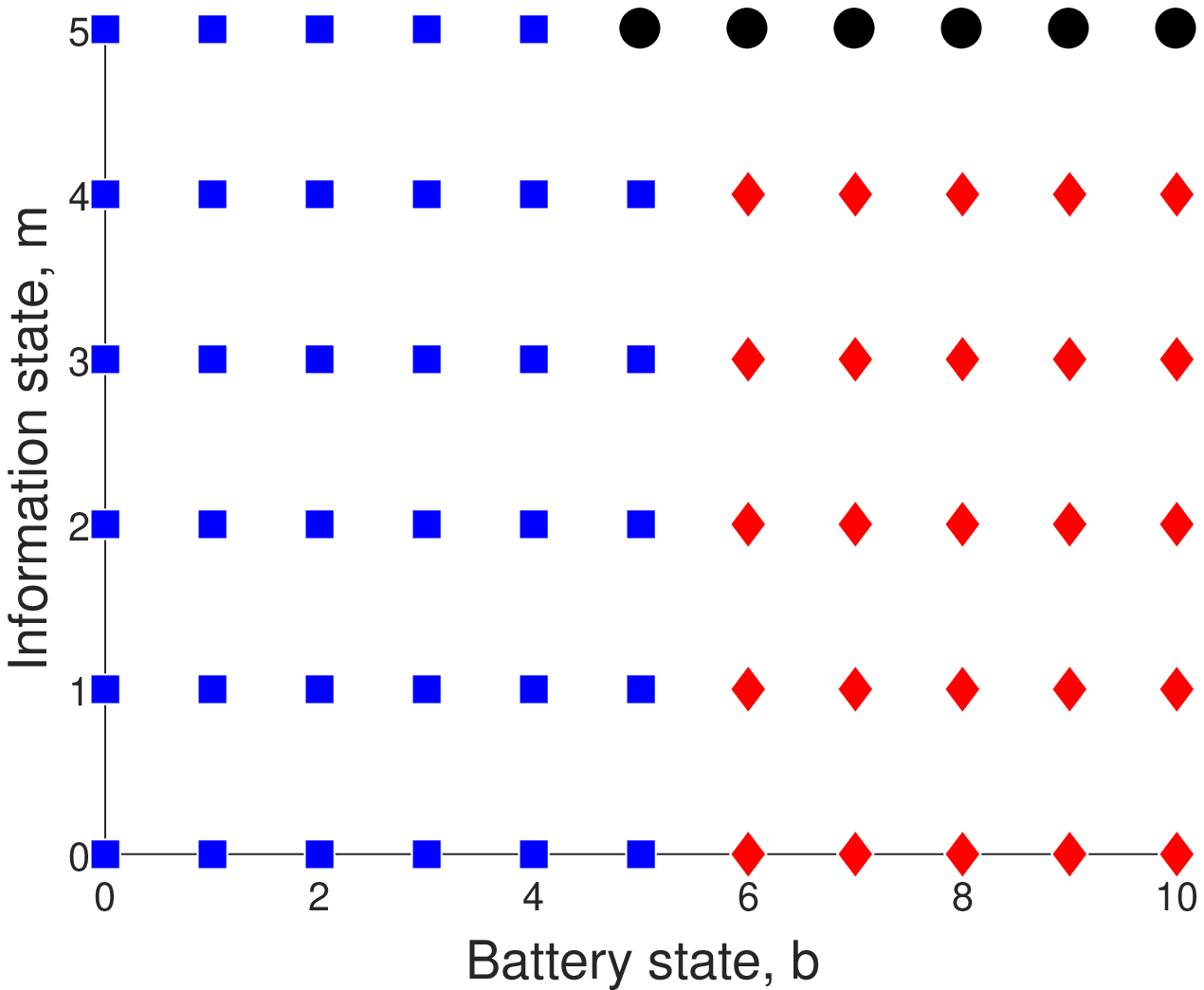}
		  \caption{Optimal $\rho$ values obtained by VIA resembling the BF policy.}
			\label{fig:BF}
\end{figure}

\begin{figure}[ht]
  \centering
    \includegraphics[scale=.5]{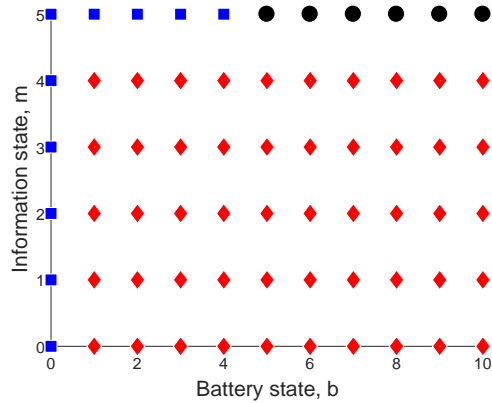}
		  \caption{Optimal $\rho$ values obtained by VIA resembling the IF policy.}
			\label{fig:IF}
\end{figure}

\section{Conclusion}\label{Conclusion}
We analyzed a point-to-point wireless link employing HARQ for reliable transmission, where the receiver can only empower itself via the transmitter's RF signal. We modeled the problem of optimal power splitting by a Markovian framework, and proved that the optimal policy should be a TS policy and as a consequence, we converted a two dimensional uncountable state Markov chain into a two dimensional countable state Markov chain. Then, we used VIA to minimize the expected number of re-transmissions and through numerical results, we showed that the optimal policy is not unique. In the future, we aim  to analytically characterize the structure of the optimal policy and to develop a low-complexity algorithm achieving the corresponding optimal performance. Additionally, we will extend the problem to the case of time-correlated channels.

\bibliographystyle{IEEEtran}
\bibliography{ref}

\end{document}